\documentclass[oribibl]{llncs}
\usepackage{a4wide}
\usepackage{amsmath}
\usepackage{amssymb}
\usepackage{graphicx}
\usepackage[noend,ruled]{algorithm2e}
\usepackage[T1]{fontenc}

\newtheorem{observation}{Observation}{\bfseries}{}
{\bfseries}{}

\newcommand{\floor}[1]{\lfloor #1\rfloor}
\newcommand{\ceil}[1]{\lceil #1\rceil}

\title{Single-Parameter Combinatorial Auctions\\ with Partially Public Valuations}

\author{Gagan Goel \and Chinmay Karande \and Lei Wang}

\institute{Georgia Institute of Technology, Atlanta.\\ \email{\{gagang, ckarande, lwang\}@cc.gatech.edu}}

\begin{document}

\date{}
\maketitle
\begin{abstract}
We consider the problem of designing truthful auctions, when the bidders' valuations have a public and a private component. In particular, we consider combinatorial auctions where the valuation of an agent $i$ for a set $S$ of items can be expressed as $v_if(S)$, where $v_i$ is a private single parameter of the agent, and the function $f$ is publicly known. Our motivation behind studying this problem is two-fold: (a) Such valuation functions arise naturally in the case of ad-slots in broadcast media such as Television and Radio. For an ad shown in a set $S$ of ad-slots, $f(S)$ is, say, the number of {\em unique} viewers reached by the ad, and $v_i$ is the valuation per-unique-viewer. (b) From a theoretical point of view, this factorization of the valuation function simplifies the bidding language, and renders the combinatorial auction more amenable to better approximation factors. We present a general technique, based on maximal-in-range mechanisms, that converts any $\alpha$-approximation non-truthful algorithm ($\alpha \leq 1$) for this problem into $\Omega(\frac{\alpha}{\log{n}})$ and $\Omega(\alpha)$-approximate truthful mechanisms which run in polynomial time and quasi-polynomial time, respectively.
\end{abstract}

\section{Introduction}

A central problem in computational mechanism design is that of combinatorial auctions, in which an auctioneer wants to sell a heterogeneous set of items $\mathcal{J}$ to interested agents. Each agent $i$ has a valuation function $v_i(.)$ which describes her valuation $v_i(S)$ for every set $S \subseteq \mathcal{J}$ of items. In its most general form, the entire valuation function is assumed to be private information which may not be revealed truthfully by the agents. Maximizing the social welfare in a combinatorial auction with an incentive-compatible mechanism is an important open problem. However, recent results \cite{DN07, sodapaper} have established polynomial lower bounds   on the approximation ratio of maximal-in-range mechanisms - which account for a majority of positive results in mechanism design - even when all the valuations are assumed to be submodular. On the other hand, in the non-game-theoretic case, if all the agents' valuations are public knowledge and hence truthfully known, then we can maximize the social welfare to much better factors \cite{DNS05, 1109675, Vondrak08}, under varying degree of restrictions on the valuations. In this paper, we introduce a model that lies in between these two extremes.

We wish to explore the setting when some inherent property of the items induces a common and publicly known {\em partial} information about the valuation function of the buyers. For instance, in  position auctions in sponsored search, the agents' valuation for a position consists of a private value-per-click as well as a public click-through rate, that is known to the auctioneer.  Another situation where such private/public factorization of valuations arises is advertisements in broadcast media such as Television and Radio. Suppose we are selling TV ad-slots on a television network. There are $m$ ad-slots and $n$ advertisers interested in them. Let us define a function $f: 2^{[m]} \rightarrow \mathbb{Z}_+$, such that for any set $S$ of ad-slots $f(S)$ is the number of {\em unique} viewers who will see the ad if the ad is shown on each slot in $S$\footnote{For a single ad-slot $j$, the function $f(\{j\})$ is nothing but the television rating for that slot as computed by rating agencies such as Nielsen. In fact, their data collection through set-top boxes results in a TV slot-viewer bipartite graph on the sample population, from which $f(S)$ can be estimated for any set $S$ of ad slots.}. If an advertiser $i$ is willing to pay $v_i$ dollars per unique viewer reached by her ad, then her total valuation of the set $S$ is $v_if(S)$. 

With this background, we define the following class of problems:

\noindent \textbf{\textsc{Single-parameter combinatorial auctions with partially public valuations}}: We are given a set $\mathcal{J}$ of $m$ items and a global \emph{public} valuation function\footnote{We do not make any explicit assumptions such as non-negativity or free disposal about the function $f$. We provide a method to convert any non-truthful black-box algorithm into a truthful mechanism. This black-box algorithm may make some implicit assumptions about $f$.} $f\ :\ 2^{\mathcal{J}} \rightarrow \mathbb{R}$. The function $f$ can either be specified explicitly or via an oracle which takes a set $S$ as input and returns $f(S)$. In addition, we have $n$ agents each of whom has a \emph{private} multiplier $v_i$ such that the item set $S$ provides $v_if(S)$ amount of utility to agent $i$. The goal is to design a truthful mechanism which maximizes $\sum_i v_if(S_i)$, where $S_1\cdots S_n$ is a partition of $\mathcal{J}$. \\

One can think of this model as combinatorial auctions with \emph{simplified bidding language}. The agents only need to specify one parameter $v_i$ as their bid. Moreover, our problem has deeper theoretical connections to the area of single-parameter mechanism design in general. For single-parameter domains such as ours, it is known that {\em monotone} allocation rules characterize the set of all truthful mechanisms. An allocation rule or algorithm is said to be monotone if the allocation parameter of an agent ($f(S_i)$ in our case) is non-decreasing in his reported bid $v_i$. Unfortunately, often it is the case that good approximation algorithms known for a given class of valuation functions are not monotonic. It is an important and well-known open question in algorithmic mechanism design to resolve whether the design of monotone algorithms is fundamentally harder than the non-monotone ones. In other words, it is not known if, for single-parameter problems, we can always convert any $\alpha$-approximation algorithm into a truthful mechanism with the same factor. We believe that our problem is a suitable candidate to attack this question as it gives a lot of flexibility in defining the complexity of function $f$. From this discussion, it follows that the only lower bound known for the approximation factor of a truthful mechanism in our setting is the hardness of approximation of the underlying optimization problem.

\subsubsection*{Our Results and techniques} 

We give a general technique which accepts any (possibly non-truthful) $\alpha$-approximation algorithm for our problem as a black-box and uses it to construct a truthful mechanism with an approximation factor of $\Omega\left(\frac{\alpha}{\log{n}}\right)$. We also give a truthful mechanism with factor $\Omega(\alpha)$ which runs in time $O\left(n^{\log\log{n}}\right)$. Both these results are corollaries obtained by setting parameters appropriately in Theorem \ref{theorem1} to achieve desired trade-off between the approximation factor and the running time. Our results can also be interpreted as converting non-monotone algorithms into monotone ones for the above model. 

Our mechanisms are {\em maximal-in-range}, \textit{i.e.}, they fix a range $\mathcal{R}$ of allocations and compute the allocation $\mathbf{S} \in \mathcal{R}$ that maximizes the social welfare. The technical core of our work lies in careful construction of this range.

While the black-box algorithm may be randomized, our mechanism does not introduce any further randomization. Depending upon whether the black-box algorithm is deterministic or randomized, our mechanism is deterministically truthful or universally truthful respectively (See Section \ref{prelim} for definitions). The approximation factor of our mechanism is deterministic (or with high probability or in expectation) if the black-box algorithm also provides the approximation guarantees deterministically (or with high probability or in expectation).

Note that we don't need to worry about how the public valuation function $f$ is specified. This is plausible since the function is accessed only from within the black-box algorithm. Hence, our mechanism can be applied to any model of specification - whether it is specified explicitly or through a value or demand oracle - using the corresponding approximation algorithm from that model. 

Submodular valuations arise naturally in practice from economies of scale or the law of diminishing returns. Hence, we make a special note of our results when the public valuation is submodular. Using the algorithm of \cite{Vondrak08} as black-box, our results imply a $\Omega\left(1/\log{n}\right)$ and $\Omega(1)$ approximation factors in polynomial time and quasi-polynomial time, respectively. We would like to note that the standard greedy algorithm for submodular welfare maximization is not monotone (See Appendix \ref{app1} for a simple example) and hence, not truthful. Similarly, the optimal approximation algorithm of \cite{Vondrak08} is also not known to be non-monotone. The best known truthful mechanism for combinatorial auctions with entirely private submodular valuations \cite{DNS05} has $\Omega(1/\sqrt{m})$ approximation factor.

\subsubsection*{Future Directions}

As shown in \cite{DN07,sodapaper}, it seems that designing a truthful mechanism with good approximation factor for maximizing social welfare is a difficult problem. In light of this, our work suggests an important research direction to pursue in combinatorial auctions- to divide the valuation function into a part which is common among all the agents and can be estimated by the auctioneer and a part which is unique and private to individual agents. 

Also, it would be interesting to see if for submodular public functions (or even more specifically, for coverage functions), which have concrete motivation in TV ad auctions, one can design a constant factor polynomial time truthful mechanism.

\subsubsection*{Related Work}

When agents have a general multi-parameter valuation function, the best known truthful approximation of social welfare in the value oracle model is $\Omega(\sqrt{\log{m}}/m)$ \cite{holzman04}. Under subadditive valuation functions, \cite{DNS05} gave $\Omega(1/\sqrt{m})$-approximate truthful mechanism. It is known that no maximal-in-range mechanism making polynomially many calls to the value oracle can have an approximation factor better than $\Omega(1/m^{1/6})$\cite{DN07} even for the case of submodular valuation functions. A similar $\Omega(1/\sqrt{m})$ hardness result for maximal-in-range algorithms based on $\mathrm{NP} \nsubseteq \mathrm{P}/\mathrm{poly}$ appears in \cite{sodapaper}. See \cite{nisan.chapter} for a comprehensive survey of the results, and \cite{PSS08, sodapaper} for other more recent work. Previous work on the single-parameter case of combinatorial auctions have primarily focused on the {\em single-minded} bidders. In this setting, any bidder $i$ is only interested in single set $S_i$ and has a valuation $v_i$ for it. Lehmann et al. \cite{LOS02} gave a truthful mechanism which achieves an essentially best-possible approximation factor of $\Omega(1/\sqrt{m})$. For other results in single-minded combinatorial auction, see \cite{MN02,APTT03}. When the desired set is publicly known and only the valuation is private, \cite{BLP09} gave a general technique which converts any $\alpha$-approximation algorithm into a truthful mechanism with factor $\alpha/\log(v_{max})$. This result is very much in spirit to our work, however the model and the techniques used in the two papers are very different. Similarly, \cite{LS05} present a general framework which uses a gap-verifying linear program as black-box to construct mechanisms that are truthful in expectation.

For the non-truthful optimization, we note that our problem is hard up to a constant factor (see \cite{MSV08}) even when all the agents have private value equal to 1 and with common valuation function being submodular. For designing monotone algorithms from non-monotone algorithms in the Bayesian setting, see \cite{HL09}. We also note that TV ad auctions are in use by Google Inc. (see \cite{Nisan09}), although currently they treat the valuations for a set of ad-slots as additive with budget constraints, which yields a multi-parameter auction.\\ 

\noindent \textbf{Organization}: Section \ref{prelim} provides a brief introduction to mechanism design with a few concepts relevant to our work. Readers familiar with design of truthful mechanisms can skip to Section \ref{basic} in which we state some basic properties and assumptions about single parameter combinatorial auctions with partially public valuations. Section \ref{vfit} introduces our vector-fitting technique and in Section \ref{simple}, we conduct a warm-up exercise by analyzing a simple mechanism. Section \ref{main} presents our main result, a vector-fitting mechanism formalized by Theorem \ref{theorem1}.

\section{Preliminaries}
\label{prelim}
In this section, we will outline the basic concepts in mechanism design relevant to our paper.

\subsection{Truthfulness and Mechanism Design}
\label{prelim1}
Mechanism design attempts to address the game-theoretic aspect of optimization problems. Let $\mathcal{A}$ be the set of alternatives, and $u_i(a)$ be the valuation of agent $i$ if alternative $a\in \mathcal{A}$ is picked. In a pure optimization setting, all the functions $u_i$'s are assumed to be known to the auctioneer, and a typical goal is to pick an alternative $a\in \mathcal{A}$ that maximizes $\sum_i u_i(a)$. But from a game-theoretic perspective, the agents may have an incentive to lie about their valuation function $u_i$, if it leads to a better alternative for them. This kind of strategizing often results in arbitrary behaviour from the agents, leading to a loss in the social welfare. Mechanism design tackles this issue by designing algorithms  such that truthfully reporting their true valuation function is the dominant strategy for each agent, \textit{i.e.} given any strategies by all the other agents, reporting one's true function maximizes the utility gained by this agent.

There are three notions of truthfulness that may be applicable:
\begin{enumerate}
\item \textbf{Deterministic truthfulness}: The mechanism must be deterministic and an agent maximizes her utility by reporting her true valuation, for any valuations of all other agents.
\item \textbf{Universal truthfulness}: A universally truthful mechanism is a probability distribution over deterministically truthful mechanisms.
\item \textbf{Truthfulness in expectation}: A mechanism is truthful in expectation if an agent maximizes her \emph{expected utility} by being truthful.
\end{enumerate}

Every deterministically truthful mechanism is universally truthful and every universally truthful mechanism is truthful in expectation. Hence, deterministic truthfulness is the strictest notion of truthfulness. As noted earlier, our mechanism may be deterministically or universally truthful depending upon whether the black-box $\alpha$-approximation algorithm is deterministic or randomized.

\subsection{Vickrey-Clarke-Grove and Maximal-in-range Mechanisms}

The Vickrey-Clarke-Grove (VCG) mechanism is a pivotal result in the field of mechanism design to maximize social welfare. It works as follows: let $a^*$ and $a^*_{-i}$ be the alternatives which maximizes $\sum_j v_j(a)$ and $\sum_{j\neq i} v_j(a)$ respectively. Now define payment $p_i$ of agent $i$ to be $\sum_{j\neq i} v_j(a^*_{-i}) - \sum_{j \neq i} v_j(a^*)$. It is now not difficult to see that with this payment function, it is in best interest of every agent to report their true valuations, irrespective of what others report.

As useful as the VCG mechanism is, it cannot be applied in many scenarios where the underlying problem is hard. Solving the optimization problem approximately doesn't preserve the truthfulness always. To overcome this, maximal-in-range variant of the VCG mechanism is a useful technique which optimizes over a smaller range of allocations. That is, the set of allocations that the mechanism may ever produce - the \emph{range} - is chosen to be a small subset of the space of all allocations. The range is chosen to balance the following trade-off: A larger range can yield better approximation but require greater computational complexity. Note that such a range needs to be defined combinatorially \textit{without} any knowledge of the agents' valuations.

For example, the $\Omega(1/\sqrt{m})$-approximate truthful mechanism from \cite{DNS05} is a maximal-in-range mechanism.

\section{Notations and Basic Properties}

\label{basic}

By boldface $\mathbf{v}$, we will denote a vector of private multipliers of the agents, where $v_i$ is the multiplier of agent $i$. For a constant $\beta \geq 0$, let $\beta\mathbf{v} = (\beta v_1, \beta v_2, ..., \beta v_n)$. By boldface $\mathbf{S}$, we will denote the vector of allocations, where $S_i$ is the set of items allocated to agent $i$. We will overload the function symbol $f$ to express the social welfare as: $f(\mathbf{v}, \mathbf{S}) = \sum_{i}{v_if(S_i)}$. An allocation $\mathbf{S}$ is \emph{optimal} for a multiplier vector $\mathbf{v}$ if it maximizes $f(\mathbf{v}, \mathbf{S})$.

We begin by observing two simple properties of our problem and its solutions: \emph{symmetry} and \emph{scale-freeness}. Our problem and its solutions are symmetric, \textit{i.e.}, invariant under relabeling of agents in the following sense: Let $\mathbf{v}$ be any multiplier vector, $\mathbf{S}$ be any allocation and $\pi$ be any permutation of $[n]$. Let $\mathbf{u}$ and $\mathbf{T}$ be such that $u_i = v_{\pi(i)}$ and $T_i = S_{\pi(i)}$. Then clearly, $f(\mathbf{v}, \mathbf{S})=\ f(\mathbf{u}, \mathbf{T})$. The problem and its solutions are also invariant under scaling, since we have $f(\beta\mathbf{v}, \mathbf{S})\ =\ \beta\cdot f(\mathbf{v}, \mathbf{S})$.

The above properties lead us to:

\begin{observation}
\label{assum1}
\emph{Without loss of generality}, every multiplier vector $\mathbf{v}$ has non-increasing entries $v_1 \geq v_2 \geq ... \geq v_n$ such that $\sum_{i}v_i = 1$.
\end{observation}

Given a multiplier vector $\mathbf{v}$, let $\mathrm{A}(\mathbf{v})$ be the optimal allocation for $\mathbf{v}$ and $\mathrm{OPT}(\mathbf{v}) = f(\mathbf{v}, \mathrm{A}(\mathbf{v}))$. Moreover, if $f(\mathbf{v}, \mathbf{S})\ \geq\ \alpha\cdot\mathrm{OPT}(\mathbf{v})$ for some $\alpha \leq 1$ then the allocation $\mathbf{S}$ is said to be $\alpha$-optimal or $\alpha$-approximate for $\mathbf{v}$.

We note a simple property of $\mathrm{A}(\mathbf{v})$: Let $\mathbf{v}$ be a multiplier vector with $v_1 \geq v_2 \geq ... \geq v_n$. Let $\mathbf{S}$ be any allocation. If $\mathbf{T}$ is a permutation of $\mathbf{S}$ such that $f(T_1) \geq f(T_2) \geq ... \geq f(T_n)$, then $f(\mathbf{v}, \mathbf{T}) \geq f(\mathbf{v}, \mathbf{S})$. In particular, if $\mathbf{S} = \mathrm{A}(\mathbf{v})$ then $f(S_1) \geq f(S_2) \geq ... \geq f(S_n)$.

Finally, we assume the existence of a poly-time black-box algorithm that computes an $\alpha$-approximate allocation $\mathrm{B}(\mathbf{v})$ for the multiplier vector $\mathbf{v}$. We express the performance guarantees of our truthful mechanisms in terms of $\alpha$ and other parameters of the problem. Although the output allocation $\mathbf{S}$ of such an algorithm may not obey $f(S_1) \geq f(S_2) \geq ... \geq f(S_n)$, it is easy to construct a non-decreasing permutation of $\mathbf{S}$ which only improves the objective function value, as discussed above. 

\begin{observation}
\label{assum2}
\emph{Without loss of generality}, any allocation $\mathbf{S}$ output by the black-box algorithm obeys $f(S_1) \geq f(S_2) \geq ... \geq f(S_n)$.
\end{observation}

Henceforth, we enforce assumptions from Observation \ref{assum1} and \ref{assum2}.

\begin{definition}[$\mathbf{u}$ dominates $\mathbf{w}$]
We say that a multiplier vector $\mathbf{u}$ dominates $\mathbf{w}$ if there exists an index $i$ such that for $k < i$, $u_k \geq w_k$ and for $k \geq i$, $u_k \leq w_k$.
\end{definition}

\begin{lemma}
\label{lemma6}
If $\mathbf{u}$ dominates $\mathbf{w}$, then $f(\mathbf{u}, \mathbf{S}) \geq f(\mathbf{w}, \mathbf{S})$ for any allocation $\mathbf{S}$ satisfying $f(S_1) \geq f(S_2) \geq ... \geq f(S_n)$.
\end{lemma}
\begin{proof}
For $k < i$, let $x_k = u_k - w_k$. Similarly for $k \geq i$, $y_k = w_k-u_k$. Then $$\sum_{k = 1}^{i-1}{x_k} - \sum_{k=i}^n{y_k}\ =\ \sum_{k=1}^n{u_k}-\sum_{k=1}^n{w_k}\ =\ 0$$ which means $\sum_{k = 1}^{i-1}{x_k} = \sum_{k=i}^n{y_k}$. Since $f(S_{k_1}) \geq f(S_{k_2})$ whenever $k_1 < i \leq k_2$,\\
$$f(\mathbf{u}, \mathbf{S}) - f(\mathbf{w}, \mathbf{S}) \ = \ \displaystyle\sum_{k=1}^{i-1}{x_kf(S_k)} - \displaystyle\sum_{k=i}^n{y_kf(S_k)} \ \geq \ 0$$\qed
\end{proof}
\noindent\textbf{Staircase Representation}: Suppose we represent a multiplier vector $\mathbf{v}$ as a histogram, which consists of $n$ vertical bars corresponding to $v_1, ..., v_n$, in that order from left to right. Since multiplier vectors have non-increasing components, such a histogram looks like a staircase descending from left to right (Refer to Figure \ref{fig4} for an example). We will refer to it as the \emph{staircase representation} of $\mathbf{v}$ and use it mainly as a visual tool.

\begin{figure}[h]
\centering
\includegraphics[height=3cm]{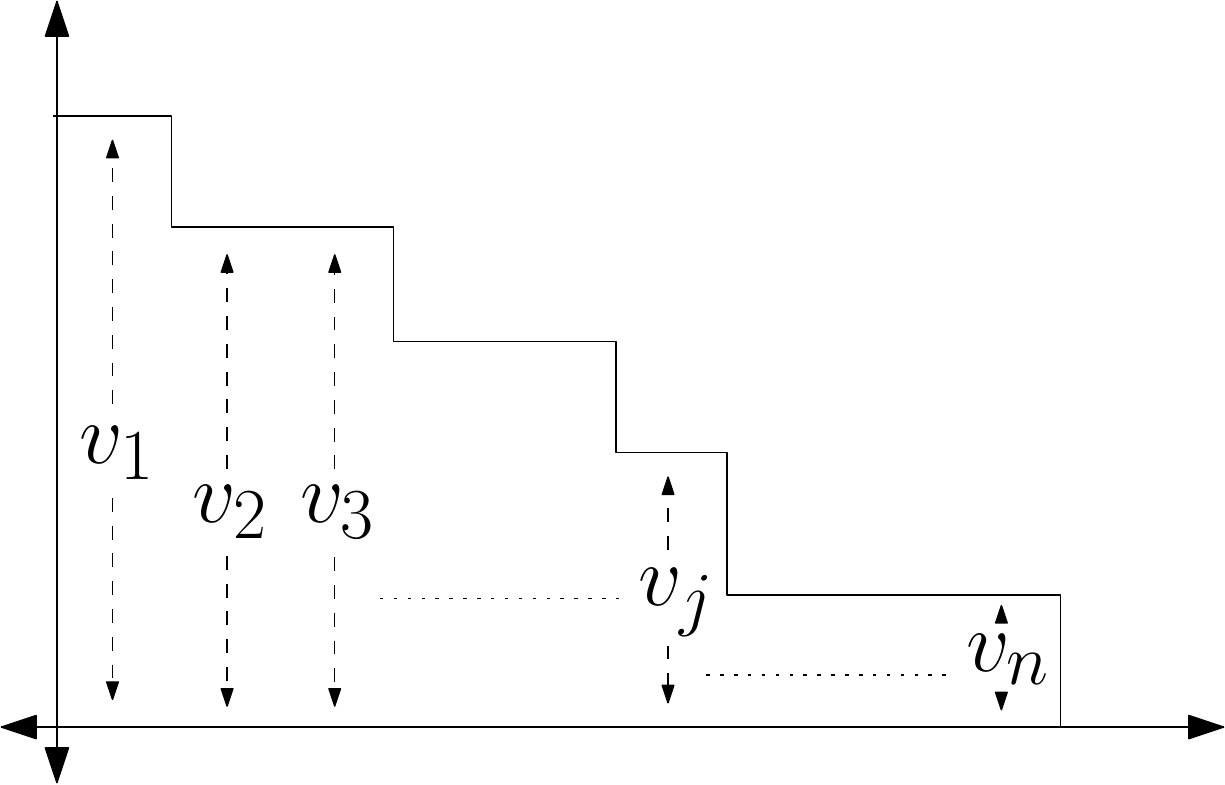}
\caption{The staircase representation of $\mathbf{v} = (v_1, ..., v_n)$.\label{fig4}}
\end{figure}

\section{Vector-Fitting Mechanisms}
\label{vfit}
Consider the following candidate approach to single parameter combinatorial auctions with partially public valuations: Fix a set $\mathcal{U}$ of some multiplier vectors. Using the black-box algorithm, compute an $\alpha$-approximate allocation $\mathrm{B}(\mathbf{v})$ for each vector $\mathbf{v} \in \mathcal{U}$ and populate the range $\mathcal{R} = \{\ B(\mathbf{v})\ :\ \mathbf{v} \in \mathcal{U}\ \}$. Run the maximal-in-range mechanism which given a multiplier vector $\mathbf{v}$, chooses the allocation $\mathbf{S} \in \mathcal{R}$ that maximizes $f(\mathbf{v}, \mathbf{S})$.

Let's consider the merits and demerits of this mechanism. If the input multiplier vector happens to be in $\mathcal{U}$, then the mechanism will indeed return an output allocation that is at least $\alpha$-approximate. But we have no guarantees otherwise. If $\mathcal{U}$ consisted of all possible vectors, we would have an $\alpha$-approximate truthful mechanism that could be computationally infeasible due to the size of $\mathcal{U}$. We handle this trade-off with \emph{vector-fitting}. The intuition behind vector-fitting is as follows: If two multiplier vectors $\mathbf{u}$ and $\mathbf{v}$ are `very similar' to each other, then $\mathrm{B}(\mathbf{u})$ and $\mathrm{B}(\mathbf{v})$ should be `similar' as well. In particular, $\mathrm{B}(\mathbf{u})$ should be a reasonably good allocation for $\mathbf{v}$ and vice versa. 

Our mechanism will be the same as the candidate mechanism outlined above, except that we will construct the set of vectors $\mathcal{U}$ very carefully. For any input vector of multipliers $\mathbf{v}$, we will guarantee that a reasonably similar vector $\mathbf{v}'$ can be found in $\mathcal{U}$, and hence and allocation $\mathbf{S}'$ is in the range $\mathcal{R}$ with provably large objective value $f(\mathbf{v}, \mathbf{S}')$.

\subsection{A Simple $\frac{\alpha}{\ln{n}}$-factor Mechanism}
\label{simple}

In this section, we will conduct a warm-up exercise by applying the vector-fitting method to construct a simple $\frac{\alpha}{\ln{n}}$-factor truthful mechanism. Recall that the vector-fitting method as outlined in Section \ref{vfit} starts with a set $\mathcal{U}$ of multiplier vectors. Our set $\mathcal{U}$ is defined as $\mathcal{U}\ =\ \{\ \mathbf{u}^j\ :\ 1 \leq j \leq n\ \}$ where $\mathbf{u}^j$ is defined as follows: 
$$u^j_i\ =\ \displaystyle\frac{1}{j}\ \ \mbox{for $1 \leq i \leq j$, zero elsewhere}$$ 

As before, for each $\mathbf{v} \in \mathcal{U}$, we compute an $\alpha$-approximate allocation $\mathrm{B}(\mathbf{v})$ and populate the range $\mathcal{R}$ with it.

Let $\mathbf{v}$ be the input multiplier vector. Let $r_j = \sum_{k=1}^j{v_j}$ for $1 \leq j \leq n$ be the prefix sums of $\mathbf{v}$. We define prefix vectors $\mathbf{d}^j$ of $\mathbf{v}$ as:
$$d^j_i\ =\ \displaystyle\frac{v_i}{r_j}\ \ \mbox{for $1 \leq i \leq j$, zero elsewhere}$$ 
It is easy to verify that $\mathbf{d}^j$ is a valid multiplier vector \textit{i.e.}, it has non-increasing components and unit $l_1$ norm.

\begin{figure}[h]
\centering
\includegraphics[height=4cm]{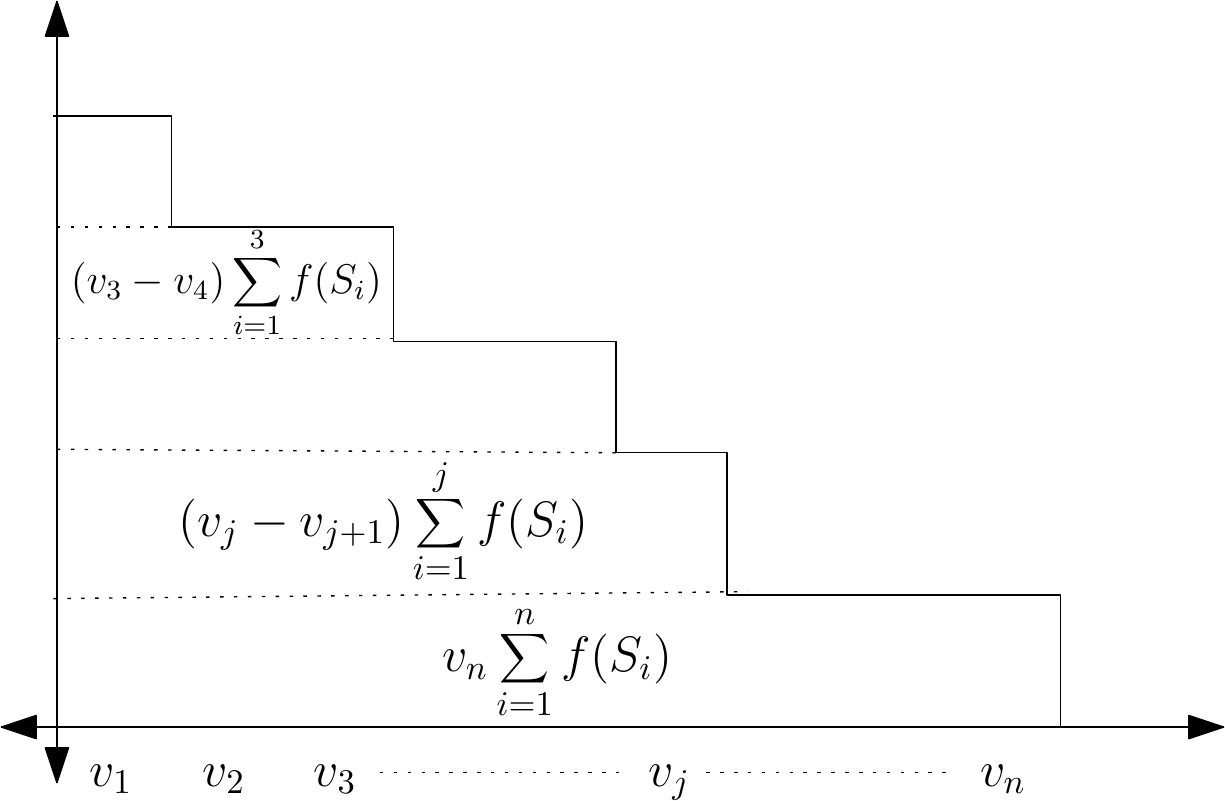}
\caption{Expressing $f(\mathbf{v}, \mathbf{S})$ as horizontal cuts of the staircase.\label{fig3}}
\end{figure}

Let $\mathbf{S} = \mathrm{A}(\mathbf{v})$ be the optimal allocation for $\mathbf{v}$ and $\mathbf{T}$ be the allocation returned by our mechanism. For notational convenience, define $v_{n+1} = 0$. We start with $\mathrm{OPT}(\mathbf{v}) = f(\mathbf{v}, \mathbf{S})$ and look at how horizontal sections under the staircase of $\mathbf{v}$ contribute to it.  See figure \ref{fig3}.

\begin{eqnarray}
\nonumber \mathrm{OPT}(\mathbf{v}) & = & \sum_{i}^n{v_if(S_i)}\\
\label{eq7} & = & v_n\sum_{i=1}^n{f(S_i)}\ +\ (v_{n-1}-v_n)\sum_{i=1}^{n-1}{f(S_i)}\ +\ ...\ +\ (v_1-v_2)f(S_1)\\
\nonumber & = &  \sum_{j = 1}^n{\left[j(v_j-v_{j+1})\sum_{i=1}^j{\frac{f(S_i)}{j}}\right]}\\
\nonumber & = & \sum_{j = 1}^n{\left[j(v_j-v_{j+1})\cdot f(\mathbf{u}^j, \mathbf{S})\right]}\\
\nonumber & \leq & \sum_{j = 1}^n{\left[j(v_j-v_{j+1})\cdot \mathrm{OPT}(\mathbf{u}^j)\right]}\\
\nonumber & \leq & \sum_{j = 1}^n{\left[\alpha^{-1} j(v_j-v_{j+1})\cdot f(\mathbf{u}^j,\mathrm{B}(\mathbf{u}^j))\right]}\\
\label{eq3} & \leq & \sum_{j = 1}^n{\left[\alpha^{-1} j(v_j-v_{j+1})\cdot f(\mathbf{d}^j,\mathrm{B}(\mathbf{u}^j))\right]}\\
\label{eq4} & \leq & \sum_{j = 1}^n{\left[\frac{j(v_j-v_{j+1})}{\alpha r_j}\cdot f(\mathbf{v},\mathrm{B}(\mathbf{u}^j))\right]}\\
\label{eq5} & \leq &  \sum_{j = 1}^n{\left[\frac{j(v_j-v_{j+1})}{\alpha r_j}\cdot f(\mathbf{v}, \mathbf{T})\right]}\\
\label{eq6} & = & \frac{f(\mathbf{v}, \mathbf{T})}{\alpha}\left[1+\sum_{j=2}^n{\frac{v_j(r_j-jv_j)}{r_jr_{j-1}}}\right]
\end{eqnarray}

Equation \eqref{eq7} decomposes $f(\mathbf{v}, \mathbf{S})$ as the horizontal cuts of the staircase of $\mathbf{v}$ (See Figure \ref{fig3}). \eqref{eq3} follows from the previous step by applying a simple structural property formalized by Lemma \ref{lemma6} to $\mathbf{d}^j$ and $\mathbf{u}^j$. Equation \eqref{eq4} follows from \eqref{eq8} below, which is a simple restatement. 

\begin{equation}
\label{eq8}
\sum_{i=1}^j{\frac{v_i}{r_j}\cdot f(S_i)}\ \leq\ \frac{1}{r_j}\sum_{i=1}^n{v_if(S_i)}\ =\ \frac{f(\mathbf{v}, \mathbf{S})}{r_j}
\end{equation}

Equation \eqref{eq6} is derived from the previous expression by simply rearranging the terms. Since $v_j \leq r_j/j$ and $r_j-jv_j \leq r_{j-1}$, we conclude that $$1+\sum_{j=2}^n{\frac{v_j(r_j-jv_j)}{r_jr_{j-1}}}\ \leq\ \sum_{j=1}^n\frac{1}{j}\ \leq\ \ln{n}$$ and $$f(\mathbf{v}, \mathbf{T}) \geq \alpha\cdot \mathrm{OPT}(\mathbf{v})/\ln{n}$$

\textbf{An example that achieves the bound}: We can differentiate each term of the summation in equation \eqref{eq6} to compute the values of $v_j$ for which the term is maximized, so as to make the bound as loose as possible. Surprisingly, a single multiplier vector maximizes all the terms simultaneously! This vector is defined as $v_j = (\sqrt{j} - \sqrt{j-1})/\sqrt{n}$. Some calculations prove that for this multiplier vector, the summation is indeed $\Omega(\ln{n})$.

\section{The Main Result}
\label{main}
In this section, we will use vector-fitting to obtain a general technique to convert a non-truthful approximation algorithm for single parameter combinatorial auctions into a truthful mechanism. This technique yields a range of trade-offs between the approximation factor and the running time of the algorithm. We will prove the following theorem:

\begin{theorem}
\label{theorem1}
There exists a truthful mechanism for maximizing welfare in a single parameter combinatorial auction with partially public valuations that runs in time $O((\log_a{n})^{\log_b{n}}\cdot \mathrm{poly}(m, n))$ and produces an allocation with total welfare at least $\frac{3\alpha}{4ab}\cdot \mathrm{OPT}(\mathbf{v})$ - where $\alpha$ is the approximation factor of the black-box optimization algorithm and $a,b > 1$ are parameters of the mechanism.
\end{theorem}

Setting $a = b = 2$ we get: (Henceforth, all logarithms are to base 2)

\begin{corollary}
\label{corr1}
There exists a $\frac{3\alpha}{16}$-factor truthful mechanism running time\\ $O((n^{\log{\log{n}}}\cdot \mathrm{poly}(m,n))$, \textit{i.e.} quasi-polynomial time.
\end{corollary}

Similarly, setting $a = 2$ and $b = \log{n}$ we get:

\begin{corollary}
\label{corr2}
There exists a truthful mechanism with factor $\Omega\left(\frac{\alpha}{\log{n}}\right)$ and polynomial running time.
\end{corollary}

When the public valuation $f$ is submodular, we have $\alpha = \left(1 - \frac{1}{e}\right)$ and the above corollaries yield factors $\Omega(1)$ and $\Omega\left(\frac{1}{\log{n}}\right)$ respectively.

\subsection{Constructing the Range $\mathcal{R}$}

\noindent \textbf{Overview}: Recall the staircase representation of a multiplier vector $\mathbf{v}$, such as in Figure \ref{fig4}. Depending upon the entries of $\mathbf{v}$, the steps of the staircase may have varying heights. We can construct a discretization of the space of all multiplier vectors by restricting the values the height of any step can take. That is, we populate the initial set $\mathcal{U}$ with all vectors whose components take values of the form $b^{-k}$ for some constant $b > 1$ and for all $k \geq 0$. Now given any input vector $\mathbf{v}$, we can find a vector $\mathbf{u} \in \mathcal{U}$ such that $u_i$ is at most a multiplicative factor $b$ away from $v_i$. Thus, $\mathbf{u}$ can serve as a vector `similar' to $\mathbf{v}$. We need more complex machinery to ensure that the size of $\mathcal{U}$ does not blow up, and that the vectors in $\mathcal{U}$ still have unit norm.

Let $a,b > 1$ be suitably chosen parameters of the mechanism. Let $Q = \{\ b^{-k}\ :\ 0\leq k < \log_b{n}\ \}$ be a set of values discretizing the interval $(\frac{1}{n}, 1]$ and $q$ be the minimum element of $Q$. For a multiplier $v_i \geq q$, we define $\floor{v_i}$ to be the largest element of $Q$ that is no greater than $v_i$. For a multiplier vector $\mathbf{v}$ we define the floor of $\mathbf{v}$, $\floor{\mathbf{v}}$ as follows:

\begin{definition}[Floor $\floor{\mathbf{v}}$]
The floor $\floor{\mathbf{v}}$ of a multiplier vector $\mathbf{v}$ is the vector $\mathbf{u}$ constructed by Algorithm \ref{alg1}.
\end{definition}

\begin{algorithm}[H]
\caption{ConstructFloor\label{alg1}}
\SetKw{KwBreak}{break}
\For{$i\ =\ 1$ to $n$}
{
	$r\ \leftarrow\ \displaystyle\frac{\left(1-\sum_{k=1}^{i-1}u_k\right)}{(n-i+1)}$;

	\BlankLine
	\BlankLine
	
	\tcc{$r$ is the minimum permissible value of $u_i$ due to monotonicity.}
	
	\BlankLine
	\BlankLine

	\If{$v_i\ \geq\ q$ and $\floor{v_i}\ >\ r$}
	{
		$u_i\ \leftarrow\ \floor{v_i}$;
	}
	\Else
	{
		\For{$j\ =\ i$ to $n$}
		{
			$u_j\ \leftarrow\ r$;
		} 
		\KwBreak
	}
}
\end{algorithm}

In short, to find the `floor' of a multiplier vector, we successively round down the `large' components into elements of $Q$, until we need to set all the remaining components equal due the monotonicity and unit norm requirement or only `small' components are remaining. When represented as a staircase (Refer Figure \ref{fig1}), all the steps of $\floor{\mathbf{v}}$ except the last one must have height that belongs to $Q$.

\begin{figure}[h]
\centering
\includegraphics[height=4cm]{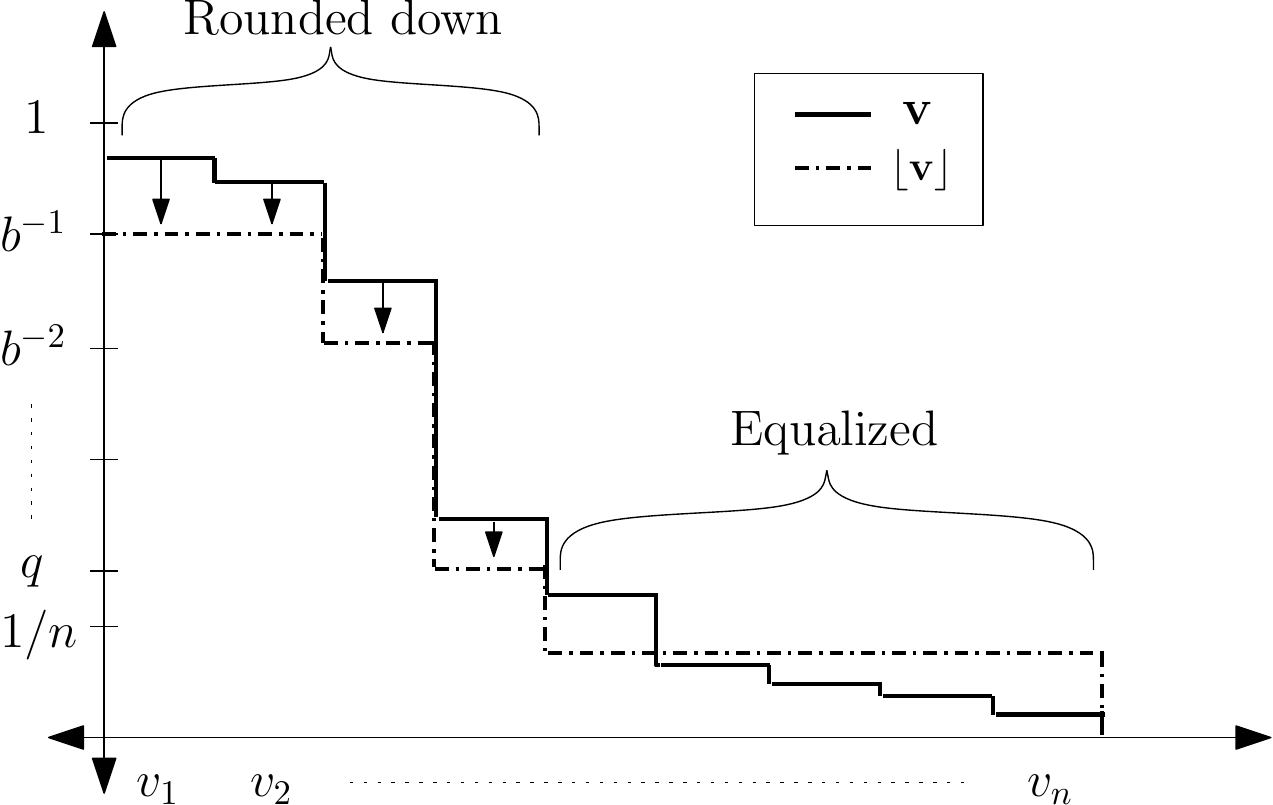}
\caption{Vertical fitting of $\mathbf{v}$.\label{fig1}}
\end{figure}

\begin{observation}
\label{obs5}
The floor of a vector $\mathbf{v}$ is a valid multiplier vector itself, \textit{i.e.} it has non-increasing components and unit $l_1$ norm. Moreover, $\mathbf{v}$ dominates $\floor{\mathbf{v}}$.
\end{observation}
\begin{proof}
Refer to Appendix \ref{app2}.\qed
\end{proof}

Intuitively, the floor of a vector is (in a sense formalized by Lemma \ref{lemma1}) `similar' to the vector, and the similarity is parametrized by $b$. 
\begin{lemma}
\label{lemma1}
For any multiplier vector $\mathbf{v}$ and allocation $\mathbf{S}$, $f(\floor{\mathbf{v}},\mathbf{S}) \geq \frac{3}{4b}\cdot f(\mathbf{v}, \mathbf{S})$.
\end{lemma}
\begin{proof}
Refer to Appendix \ref{appendix1}.\qed
\end{proof}

We will construct our preliminary set of vectors $\mathcal{U}'$ as $$\mathcal{U}'\ =\ \{\ \mathbf{u}\ :\ \mathbf{u} = \floor{\mathbf{v}}\ \mbox{for some multiplier vector}\ \mathbf{v}\ \}$$ It turns out that $\mathcal{U}'$ is too large for our purposes. Hence we construct a subset $\mathcal{U} \subseteq \mathcal{U}'$, which is small enough. Referring back to the staircase representation of a multiplier vector (Figure \ref{fig1}), we constructed $\mathcal{U}'$ by discretizing the `height' of each step - by fitting the vectors vertically. Since rounding down the components of $\mathbf{v}$ might lead to many components of $\mathbf{u} = \floor{\mathbf{v}}$ having the same value, $\mathbf{u}$ also looks like a staircase, perhaps with `wider' steps. Each step of $\mathbf{u}$ may have any integral width - at most $n$. 

We construct $\mathcal{U}$ from $\mathcal{U}'$ by further restricting how wide a step can be - by horizontal fitting (See Figure \ref{fig2}). We allow each step (except the last) to be of width $\ceil{a^k}$ for some integer $k \geq 0$ - where $a > 1$ is a suitably chosen parameter of the mechanism. To this end, we need to slightly formalize the staircase representation of a multiplier vector, which till now we only used as a visual aid. By a \emph{step} of the staircase of $\mathbf{v}$, we will mean a maximal interval $[i_1, ..., i_2] \subseteq [1, ..., n]$ such that $v_{i_1} = v_{i_2}$. All the indices $i_1 \leq i \leq i_2$ will be said to belong to the step, whereas $i_1$ and $i_2$ and the first and last indices of the step. The \emph{height} of the step is given by $v_{i_1}$ and the \emph{width} by $i_2-i_1+1$.

\textbf{Remark}: Notice that just as a multiplier vector can be specified by the $n$-tuple $(v_1, ..., v_n)$, it can also be identified by specifying the height and width of each step of its staircase representation. In fact, specifying all but the last step of a staircase fixes the last step due to the unit norm requirement.

For a multiplier vector $\mathbf{v}$, we define the core $\overleftarrow{\mathbf{v}}$ of $\mathbf{v}$ as:

\begin{definition}[Core $\overleftarrow{\mathbf{v}}$]
The core $\overleftarrow{\mathbf{v}}$ of a multiplier vector $\mathbf{v}$ is the vector $\mathbf{u}$ constructed by Algorithm \ref{alg2}.
\end{definition}

\begin{algorithm}[H]
\caption{ConstructCore \label{alg2}}
\SetKw{KwBreak}{break}
$i_1 \leftarrow 1$; $j_1 \leftarrow 1$;

\While{$i_1 \leq n$}
{
	$r \leftarrow \left(1-\sum_{i=1}^{j_1-1}{u_i}\right)/(n-j_1+1)$;

	\BlankLine
	
	\If{$v_{i_1} > r$}
	{
		Find the largest index $i_2$ such that $v_{i_1} = v_{i_2}$;		
		\BlankLine		
		Find largest integer $k$ such that $\ceil{a^k} \leq (i_2-j_1+1)$;
		
		\For{$i\ =\ j_1$ to $j_1+\ceil{a^k}-1$}
		{
			$u_i\ \leftarrow\ v_{i_1}$;
		}
		
		$i_1 \leftarrow i_2+1$;
		
		$j_1 \leftarrow j_1+\ceil{a^k}$;
	}
	\Else
	{
		\For{$i\ =\ j_1$ to $n$}
		{
			$u_i\ \leftarrow\ r$;
		}
		\KwBreak
	}
	
}
\end{algorithm}

\noindent \textbf{Operation of Algorithm \ref{alg2}}: Each iteration of the \texttt{while} loop processes one step of $\mathbf{v}$ and $\mathbf{u}$. $i_1$ and $j_1$ hold the first index of the current step of $\mathbf{v}$ and $\mathbf{u}$ respectively. $r$ is the minimum height of the current step of $\mathbf{u}$ by monotonicity. If $r \geq v_{i_1}$, then the requirement for unit $l_1$ norm forces us to introduce the last step of the staircase of $\mathbf{u}$. Otherwise, $[i_1, ..., i_2]$ is the current step of $\mathbf{v}$ and we set the width of the current step of $\mathbf{u}$ to be $\ceil{a^k}$.

\begin{figure}[h]
\centering
\includegraphics[height=4cm]{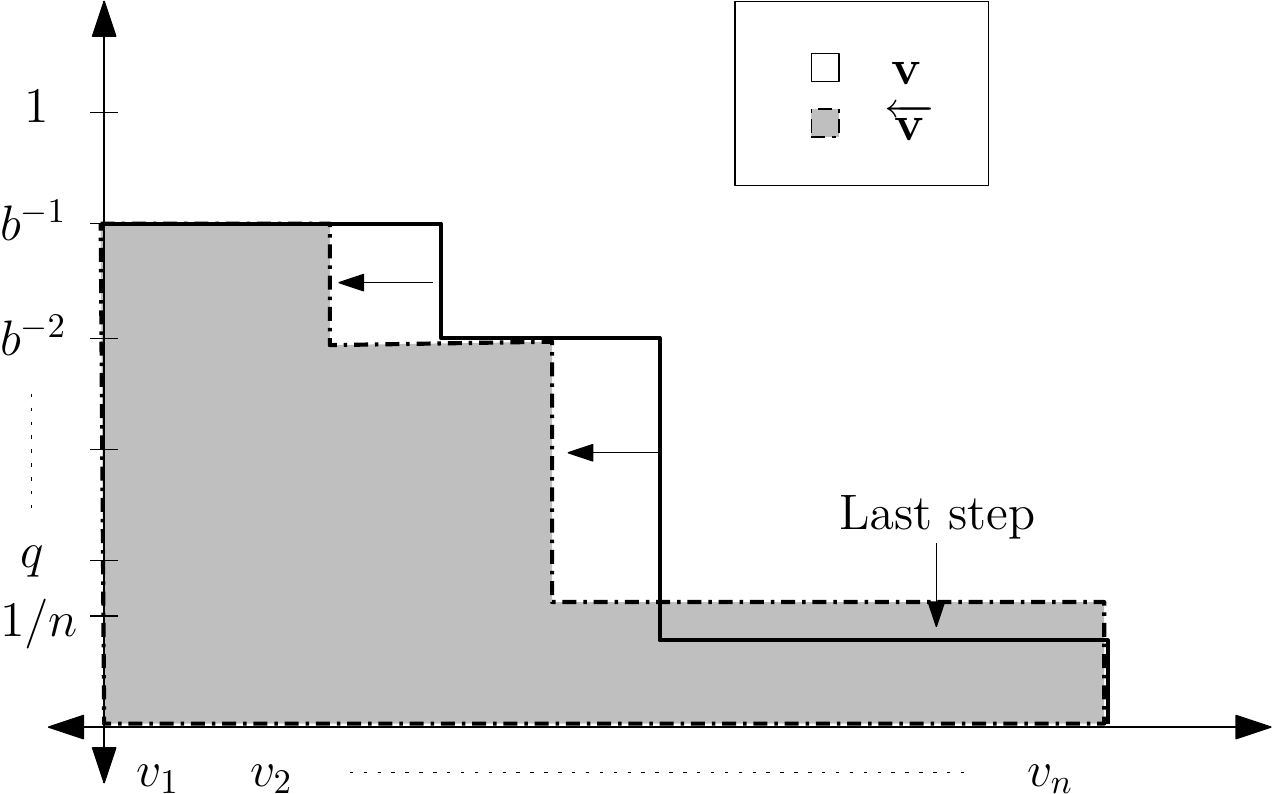}
\caption{Horizontal fitting of $\mathbf{v}$.\label{fig2}}
\end{figure}

\begin{observation}
\label{obs6}
The core of a vector $\mathbf{v}$ is a multiplier vector itself, \textit{i.e.} it has non-increasing components and unit $l_1$ norm. Moreover, $\mathbf{v}$ dominates $\overleftarrow{\mathbf{v}}$.
\end{observation}
\begin{proof}
Refer to Appendix \ref{app3}.\qed
\end{proof}

\begin{lemma}
\label{lemma2}
For any multiplier vector $\mathbf{v}$ and allocation $\mathbf{S}$,  $f(\overleftarrow{\mathbf{v}}, \mathbf{S})\ \geq\ f(\mathbf{v}, \mathbf{S})/a$.
\end{lemma}
\begin{proof}
Refer to Appendix \ref{appendix2}.\qed
\end{proof}

We now define our set of vectors $\mathcal{U}$ as follows: $ \mathcal{U}\ =\ \{\ \overleftarrow{\mathbf{v}}\ :\ \mathbf{v} \in \mathcal{U}'\ \}$. We populate the range $\mathcal{R}$ of allocations as $\mathcal{R}\ =\ \{\ \mathrm{B}(\mathbf{v})\ :\ \mathbf{v} \in \mathcal{U}\ \}$ where $\mathrm{B}(\mathbf{v})$ is the $\alpha$-approximate allocation returned by the black box algorithm.

\subsection{Proof of Theorem \ref{theorem1}}

We run the following maximal-in-range mechanism: Given an input multiplier vector $\mathbf{v}$ we return the allocation $\mathbf{T} \in \mathcal{R}$ that maximizes $f(\mathbf{v}, \mathbf{T})$. We need to prove that $f(\mathbf{v}, \mathbf{T})\ \geq\ \frac{3\alpha}{4ab}\cdot \mathrm{OPT}(\mathbf{v})$

Let $\mathbf{S} = \mathrm{A}(\mathbf{v})$ be the optimal allocation for $\mathbf{v}$ and $\overleftarrow{\floor{\mathbf{v}}}$ be \emph{the core of the floor} of $\mathbf{v}$. Combining Lemmas \ref{lemma1} and \ref{lemma2}, we conclude that $f(\overleftarrow{\floor{\mathbf{v}}}, \mathbf{S})\ \geq\ \frac{3}{4ab}\cdot \mathrm{OPT}(\mathbf{v})$. Since $\overleftarrow{\floor{\mathbf{v}}} \in \mathcal{U}$, there exists an allocation $\mathbf{X} \in \mathcal{R}$ such that 
\begin{equation}
\label{eq1}
f(\overleftarrow{\floor{\mathbf{v}}}, \mathbf{X})\ \geq\ \alpha\cdot \mathrm{OPT}(\overleftarrow{\floor{\mathbf{v}}})\ \geq\ \alpha\cdot f(\overleftarrow{\floor{\mathbf{v}}}, \mathbf{S})\ \geq\ \frac{3\alpha}{4ab}\cdot \mathrm{OPT}(\mathbf{v})
\end{equation}

Since $\mathbf{v}$ dominates $\floor{\mathbf{v}}$ which in turn dominates $\overleftarrow{\floor{\mathbf{v}}}$ (Refer to Observation \ref{obs5} and \ref{obs6}), application of Lemma \ref{lemma6} yields:
\begin{equation}
\label{eq2}
f(\mathbf{v}, \mathbf{X})\ \geq\ f(\floor{\mathbf{v}}, \mathbf{X})\ \geq\ f(\overleftarrow{\floor{\mathbf{v}}}, \mathbf{X})
\end{equation}

Using equations \eqref{eq1} and \eqref{eq2}, $$f(\mathbf{v}, \mathbf{T})\ \geq\ f(\mathbf{v}, \mathbf{X})\ \geq\ f(\overleftarrow{\floor{\mathbf{v}}}, \mathbf{X})\ \geq\ \frac{3\alpha}{4ab}\cdot \mathrm{OPT}(\mathbf{v})$$ The running time of the mechanism is established by Lemma \ref{lemma5}, which finishes the proof of Theorem \ref{theorem1}.

\begin{lemma}
\label{lemma5}
$|\mathcal{R}|\ =\ O\left((\log_a{n})^{\log_b{n}}\right)$
\end{lemma}
\begin{proof}
$|\mathcal{R}|$ is bounded by $|\mathcal{U}|$. $\mathcal{U}$ consists of only those vectors which are cores of floors of some multiplier vectors. We have seen that each step of the staircase of $\mathbf{v} \in \mathcal{U}$ except the last must be of width $w = \ceil{a^k}$ for some integer $k$. Moreover, there can be only $|Q| = O(\log_b{n})$ such steps and at most one of each height. We have also remarked that specifying all but the last step of a staircase fixes it. Therefore there can be at most $O\left((\log_a{n})^{\log_b{n}}\right)$ distinct staircases in $\mathcal{U}$.\qed
\end{proof}

\bibliography{GKW_auctions}
\bibliographystyle{plain}

\appendix

\section{The Greedy Algorithm is not Monotone}
\label{app1}

The greedy algorithm in our setting works as follows: In each step it assigns one unallocated item $j$ to a buyer $i$, where the pair $(i,j)$ is chosen so as maximize the marginal gain in the objective function. That is, if buyer $i$ had been allocated the set $S$ of items before the current step, then $v_i\left(f(S\cup \{j\})-f(S)\right)$ is maximized.

We will construct an example that adheres to our formulation of the TV Ad auctions problem. Consider an instance with two advertisers $i_1$ and $i_2$ and three ad-slots $j_1, j_2, j_3$. Suppose there are 10 viewers $k_1, ..., k_{10}$. Viewers $k_1$ to $k_5$ watch slot $j_1$, $k_6$ to $k_{10}$ watch slot $j_2$ and $k_3$ to $k_8$ watch slot $j_3$. The public function $f$ in this case is the coverage function: for a set $S$ of slots $f(S)$ is the number of unique viewers who watch any slot in $S$. To prove that the greedy algorithm does not make monotonic allocations in this example, consider two cases:

\begin{enumerate}
\item \textbf{$v_{i_1} = 1$ and $v_{i_2} = 1+\epsilon$}: In the first step, the greedy algorithm assigns the largest slot $j_3$ (with six viewers) to $i_2$. In the next two steps, it assigns both $j_1$ and $j_2$ to $i_1$. Therefore, the $i_1$ receives the set $\{j_3\}$ of total allocation value (not counting the private multiplier) 6. 
\item \textbf{$v_{i_1} = 1$ and $v_{i_2} = 1-\epsilon$}: In the first step, the greedy algorithm assigns the largest slot $j_3$ (with six viewers) to $i_1$. In the next two steps, it assigns both $j_1$ and $j_2$ to $i_2$. Therefore, the $i_1$ receives the set $\{j_1,j_2\}$ of total allocation value (not counting the private multiplier) 10. 
\end{enumerate}

Clearly, $i_2$ receives a larger allocation at a lower private valuation. Therefore, the greedy algorithm is not monotone.

\section{Proof of Observation \ref{obs5}}
\label{app2}
The procedure to compute $\mathbf{u} = \floor{\mathbf{v}}$ easily ensures the unit $l_1$ norm. Now to prove monotonicity by contradiction, assume that there exists $i$ such that $u_i < u_{i+1}$. Since $\mathbf{v}$ satisfies monotonicity, this can only happen if $v_i$ was strictly rounded down to $u_i$ and $v_{i+1}$ was not. Therefore $$1-\sum_{k = 1}^{i-1}{u_k}\ =\ \sum_{k=i}^n{u_k}\ =\ u_i + \sum_{k=i}^n{u_k} \ =\ u_i + (n-i)u_{i+1}>\ (n-i+1)\cdot u_i$$ This implies $$\floor{v_i}\ =\ u_i\ <\ \frac{1-\sum_{k = 1}^{i-1}{u_k}}{n-i+1}$$ which is impossible since the right-hand side of the above inequality is the minimum value $u_i$ could have been assigned. 

To see that $\mathbf{v}$ dominates $\mathbf{u} = \floor{\mathbf{v}}$, observe that if the index $i$ does not belong to the last step of $\mathbf{u}$, then $v_i$ must have been rounded down to $u_i$, and therefore, $u_i \leq v_i$. Now consider the smallest $i$ such that $u_i > v_i$. Then $i$ must belong to the last step of $\mathbf{u}$, and hence $u_j = u_i > v_i \geq v_j$ for any $j \geq i$.\qed

\section{Proof of Lemma \ref{lemma1}}
\label{appendix1}

Define $\mathbf{u} = \floor{\mathbf{v}}$. Let $p$ be the highest index such that $v_p$ is rounded down by the procedure that constructs $\mathbf{u}$, \textit{i.e.} $u_p = \floor{v_p}$ and $u_p > r = u_{p+1}$. Since, $\sum_{i=1}^p{u_i}\ \leq\ \sum_{i=1}^p{v_i}$, it is clear that $p < n$. Now for $i \leq p$, we have $u_i=\floor{v_i}\geq v_i/b$. Consider two cases about $v_{p+1}$:

\textbf{Case 1 - $v_{p+1} \geq q$}: In this case, $u_{p+1}\ =\ r\ \geq\ \floor{v_{p+1}}\ \geq\ v_{p+1}/b$. For $i \geq p+1$, we have $v_i \leq v_{p+1}$ and $u_i = u_{p+1}$ implying $u_i \geq v_i/b$. Therefore, $$f(\mathbf{u}, \mathbf{S})\ =\ \sum_{i=1}^n{u_if(S_i)}\ \geq\ \frac{1}{b}\sum_{i=1}^n{v_if(S_i)}\ =\ \frac{1}{b}\cdot f(\mathbf{v}, \mathbf{S})$$

\textbf{Case 2 - $v_{p+1} < q$}: Let $h = \displaystyle\sum_{i=1}^pv_i$ and $H = \left(\displaystyle\sum_{i=1}^p{v_if(S_i)}\right)/f(\mathbf{v}, \mathbf{S})$. From the monotonicity of $\mathbf{S}$, we conclude that $$H\cdot f(\mathbf{v}, \mathbf{S})\ =\ \sum_{i=1}^p{v_if(S_i)}\ \geq\ h\cdot f(\mathbf{v}, \mathbf{S})$$ and hence $H \geq h$.

Since $u_i \leq v_i$ for all $i \leq p$, and both $\mathbf{u}$ and $\mathbf{v}$ must have unit $l_1$ norm, we have $\sum_{i>p}{u_i}\ \geq\ \sum_{i > p}{v_i}\ =\ (1-h)$. Hence, $u_i \geq \frac{1-h}{n}$ for $i > p$. By definition, $v_i < q \leq \frac{b}{n}$ for $i > p$. Together, these imply $u_i \geq (1-h)v_i/b$. Finally, using $H \geq h$, we conclude $$\sum_{i>p}{u_if(S_i)}\ \geq\ \frac{1-h}{b}\left(\sum_{i>p}{v_if(S_i)}\right)\ \geq\ \frac{1-H}{b}\left[(1-H)f(\mathbf{v}, \mathbf{S})\right]$$

Combining these pieces together, we get:
\begin{eqnarray}
\nonumber f(\floor{\mathbf{v}},\ \mathbf{S}) & = & \sum_{i=1}^p{u_if(S_i)}\ +\ \sum_{i>p}{u_if(S_i)}\\
\nonumber & \geq & \frac{1}{b}\sum_{i=1}^p{v_if(S_i)}\ +\ \frac{(1-H)^2}{b}\cdot f(\mathbf{v}, \mathbf{S})\\
\nonumber & = & \frac{H+(1-H)^2}{b}\cdot f(\mathbf{v}, \mathbf{S}) \ \geq \ \frac{3}{4b}\cdot f(\mathbf{v}, \mathbf{S})
\end{eqnarray}\qed

\section{Proof of Observation \ref{obs6}}
\label{app3}

The algorithm to construct $\mathbf{u} = \overleftarrow{\mathbf{v}}$ itself easily ensures the unit norm. To prove monotonicity by contradiction, assume that there exists $i$ such that $u_i < u_{i+1}$. This can only happen is $i+1$ is the first index of the last step of $\mathbf{u}$ and $i$ is the last index of the penultimate step. Let $j$ be the first index of the penultimate step. Then $$1-\sum_{k = 1}^{j-1}{u_k}\ =\ \sum_{k = j}^n{u_k}\ =\ (i-j+1)u_i + (n-i)u_{i+1}\ >\  (n-j+1)\cdot u_i$$ This means $$u_j\ =\ u_i\ <\ \left(1-\sum_{k = 1}^{j-1}{u_k}\right)/(n-j+1)$$ which is impossible since the right-hand side of the above inequality is the minimum value $u_j$ could have been assigned.

To see that $\mathbf{v}$ dominates $\mathbf{u} = \overleftarrow{\mathbf{v}}$, observe that if the index $i$ does not belong to the last step of $\mathbf{u}$, then $u_i = v_j$ for some $j \geq i$, and hence $u_i = v_j \leq v_i$. Now consider the smallest $i$ such that $u_i > v_i$. Then $i$ must belong to the last step of $\mathbf{u}$. Therefore, $u_j = u_i > v_i \geq v_j$ for all $j \geq i$.\qed

\section{Proof of Lemma \ref{lemma2}}
\label{appendix2}

Suppose the staircase of $\mathbf{v}$ has $s_1$ steps and that of $\mathbf{u} = \overleftarrow{\mathbf{v}}$ has $s_2$ steps. Then the following four properties follow directly from the algorithm:
\begin{enumerate}
\item $s_2 \leq s_1$
\item For $1 \leq i < s_2$, the $i$'th step of $\mathbf{v}$ is at most $a$ times as wide as the $i$'th step of $\mathbf{u}$ and both have the same height.
\item For $1 \leq i \leq s_2$, let $i_1$ and $j_1$ be the first indices of the $i$'th steps of $\mathbf{v}$ and $\mathbf{u}$ respectively. Then $i_1 \geq j_1$.
\item If $[j, ..., n]$ is the last step of $\mathbf{u}$ then $u_i \geq v_i$ for $i \geq j$.
\end{enumerate}

To prove the lemma, we will compare the the contributions of corresponding steps of the staircases of $\mathbf{v}$ and $\mathbf{u}$ to the objective functions.

For $i < s_2$, let $[i_1, ..., i_2]$ be the $i$'th step of $\mathbf{v}$, $[j_1, ..., j_2]$ be the $i$'th step of $\mathbf{u}$ and $h = v_{i_1} = u_{j_1}$ be their common height. We have $$\sum_{k = j_1}^{j_2}{u_kf(S_k)}\ =\ h\sum_{k=j_1}^{j_2}{f(S_k)}\ \geq\ h\sum_{k=i_1}^{i_i+j_2-j_1}{f(S_k)}$$ by the third property. The monotonicity of $\mathbf{S}$ and the second property then imply $$\sum_{k = j_1}^{j_2}{u_kf(S_k)}\ \geq\ \frac{1}{a}\sum_{k = i_1}^{i_2}{v_kf(S_k)}$$ So the $i$'th step of $\mathbf{v}$ contributes at most $a$ times value to $f(\mathbf{v}, \mathbf{S})$ as the $i$'th step of $\mathbf{u}$ contributes to $f(\mathbf{u}, \mathbf{S})$, where $i < s_2$.

Finally by the fourth property, the step $s_2$ of $\mathbf{u}$ contributes more to $f(\mathbf{u}, \mathbf{S})$ than the corresponding contribution of steps $s_2, ..., s_1$ of $\mathbf{v}$ to $f(\mathbf{v}, \mathbf{S})$ combined. The result therefore follows.\qed

\end{document}